\documentclass{jams-l}
\newcommand{\R}{\mathbb{R}}

\usepackage{enumerate}
\usepackage{graphicx}
\usepackage{xcolor}

\usepackage{hyperref}       

\newtheorem{theorem}{Theorem}
\newtheorem*{theorem*}{Theorem}
\newtheorem*{claim*}{Claim}
\newtheorem{lemma}[theorem]{Lemma}

\newtheorem{claim}[theorem]{Claim}

\newtheorem*{defn}{Definition}
\newtheorem*{remark}{Remark}

\newcommand{\Z}{\mathbb{Z}}
\newcommand{\C}{\mathbb{C}}
\newcommand{\disc}{\mathsf{disc}}

\numberwithin{equation}{section}

\begin{document}

\title{The discrepancy of greater-than}



\author{Srikanth Srinivasan}
\address{Department of Computer Science, University of Copenhagen
and Department of Computer Science, Aarhus University}
\email{srsr@di.ku.dk}

\author{Amir Yehudayoff}
\address{Department of Computer Science, University of Copenhagen
and Department of Mathematics, Technion-IIT}
\email{amir.yehudayoff@gmail.com}

\begin{abstract}
The discrepancy of the $n \times n$
greater-than matrix 
is shown to be $\frac{\pi}{2 \ln n}$ up to lower order terms. 
\end{abstract}

\maketitle

The greater-than matrices 
appear in many areas of mathematics.
They form a central example in communication complexity (see e.g.~\cite{rao2020communication}).
They are studied in analysis as the main triangle projection (e.g.~\cite{kwapien1970main}).
They serve as a model for threshold gates in circuit complexity (e.g.~\cite{hajnal1993threshold}).
Understanding their properties is therefore a fundamental problem.

\begin{defn}
Let $G = G_n$ be the $n \times n$ signed greater-than matrix:
for every $j,k \in [n]$,
$$G_{j,k} = 2_{j+k \leq n}-1.$$
\end{defn}

One way to capture the structure of an object is using discrepancy~\cite{matousek1999geometric,chazelle1998discrepancy}.
On a high-level, discrepancy measures the maximum correlation 
with certain test functions, and
it captures pseudo-randomness properties. 
Discrepancy of matrices plays a key role in communication complexity.
It allows to lower bound randomized communication complexity
(see~\cite{rao2020communication} and references within), and
it allows to bound information complexity~\cite{braverman2016discrepancy}.
It also satisfies a direct product property~\cite{shaltiel2001towards,lee2008direct}.

\begin{defn}
The discrepancy of an $n \times n$ real-valued matrix $M$ with respect to a distribution $\mu = \mu_{j,k}$
on $[n] \times [n]$ is
$$\disc_\mu(M) = \max_{x,y \in \C^n: \|x\|_\infty = \|y\|_\infty = 1} 
\Big| \sum_{j,k} M_{j,k} \mu_{j,k} x_j y_k \Big| .$$
The discrepancy of $M$ is
$$\disc(M) = \inf_{\mu} \disc_{\mu}(M).$$
\end{defn}

\begin{remark}
In communication complexity, the standard definition of discrepancy uses
boolean vectors $x,y$ instead of complex vectors.
Our proof leads to a sharp bound for the complex version
(which is equal to the boolean version up to constant factors). 
\end{remark}

The main result of this note is a sharp analysis of the discrepancy
of the greater-than matrix.

\begin{theorem*}
$$\disc(G_n)  = (1-o(1)) \frac{\pi}{2 \ln n}.$$
\end{theorem*}

This sharp bound improves the previous upper bound $\disc(G_n) \leq O(\tfrac{1}{\sqrt{\ln n}})$ proved by Braverman and Weinstein~\cite{braverman2016discrepancy}.
It also improves all previous lower bounds
on the two-party public-coin communication complexity 
of the greater-than function~\cite{viola2015communication,braverman2016discrepancy,ramamoorthy2015communication}.

To prove an upper bound on $\disc(G)$,
we need to choose an appropriate distribution $\mu$.
The ``natural'' distributions---that were used by Viola~\cite{viola2015communication}, by Braverman and Weinstein~\cite{braverman2016discrepancy},
and by Ramamoorthy and Sinha~\cite{ramamoorthy2015communication}---lead 
to sub-optimal discrepancy $\Omega(\tfrac{1}{\sqrt{\ln n}})$;
for more details see Section~\ref{sec:previous}.
The distribution we use is constructed via a Hilbert matrix,
and it is based on ideas of Kwapien and Pelczyski~\cite{kwapien1970main}
and of Titchmarsh~\cite{titchmarsh1926reciprocal}
from analysis.

To prove a lower bound on $\disc(G)$,
we need to identify the witness vectors $x,y$.
There is a natural and simple mechanism for locating $x,y$ that yields
an $\Omega(\tfrac{1}{\ln n})$ lower bound on the discrepancy (see e.g.~\cite{avraham2022blocky}).
Getting an exact bound, however, is not so simple.
We use deep ideas of Bennett~\cite{Bennett_1977} from the study of Schur multipliers.
Somewhat surprisingly, the mechanism that enables to locate the witnesses $x,y$ uses abstract machinery,
like the Hahn-Banach theorem, the Riesz representation theorem
and the F.\ and M.\ Riesz theorem.

\section{The upper bound}

Fix $n$ for the rest of this text, and
let $H = H_n$ be the following $n \times n$ version of the Hilbert matrix:
\begin{align}
\label{eqn:H}
H_{j,k} = 
\frac{1}{n+\tfrac{1}{2}-j-k} .
\end{align}
This Hilbert matrix has three useful properties
that are described in the following three claims.

\begin{claim}
$H_{j,k} G_{j,k} \geq 0$ for all $j,k$.
\end{claim}

\begin{claim}[follows~\cite{titchmarsh1926reciprocal}]
\label{clm:titch}
$| x H y | \leq \pi \|x\|_2 \|y\|_2$
for all $x,y \in \C^n$.
\end{claim}

\begin{claim}
\label{clm:Hlarge}
$\|H\|_1 \geq 2 n (\ln(n)-3) $.
\end{claim}

\begin{proof}[Proof of upper bound using the three claims]
Define a distribution $\mu^*$ on the entries by
$$\mu^*_{j,k} = \frac{|H_{j,k}|}{\|H\|_1}.$$  
Let $x,y \in \C^n$ be the vectors of $\ell_\infty$-norm one that
witness the discrepancy of $G$ with respect to $\mu^*$.
The $\ell_2$-norm of $x,y$ is at most $\sqrt{n}$.
Think of $x$ as a row vector and of $y$ as a column vector.
Bound
\begin{equation*}
\disc_{\mu^*}(G) = \frac{1}{\|H\|_1} |x H y| 
 \leq \frac{n}{\|H\|_1} \Big|\tfrac{x}{\|x\|_2} H \tfrac{y}{\|y\|_2}\Big|  \leq \frac{\pi}{2 (\ln(n)-3) } .
\qedhere 
\end{equation*}
\end{proof}

\begin{proof}[Proof of Claim~\ref{clm:titch}]
Because $H$ is symmetric, we need to upper bound
its spectral norm over $\R$: for every $x \in \R^n$,
$$\|Hx\|_2 \leq \pi \|x\|_2.$$
To prove this, extend $H$ to a larger matrix.
Let $N$ be a large integer and let $J = [-N,N] \cap \Z$.
Let $\tilde H$ be the $J \times J$ 
matrix defined by the formula in ~\eqref{eqn:H}.
Let $\tilde x \in \R^{J}$ be so that
$\tilde x_j = x_j$ for every $j \in [n]$,
and $\tilde x_j = 0$ for every $j \in J \setminus [n]$.
Bound
\begin{align*} 
\|H x\|_2^2 
& \leq \|\tilde H \tilde x \|_2^2 \\
& = \sum_{j \in J} \Big| \sum_{k \in J} \frac{1}{n+\tfrac{1}{2}-j-k} \tilde x_k \Big|^2 \\
& = \sum_j \Big( \sum_{k} \frac{1}{(n+\tfrac{1}{2}-j-k)^2} \tilde x_k^2 \Big)
\\ 
& \qquad + \Big( \sum_{k_1 \neq k_2} \frac{1}{n+\tfrac{1}{2}-j-k_1}
\frac{1}{n+\tfrac{1}{2}-j-k_2} \tilde x_{k_1} {\tilde x_{k_2}} \Big) \\
& = \Big( \sum_k \tilde x_k^2 \sum_{i } \frac{1}{(n+\tfrac{1}{2}-j-k)^2}  \Big)
\\ 
& \qquad  + \Big( \sum_{k_1 \neq k_2} \tilde x_{k_1} {\tilde x_{k_2}}
\sum_{j} \frac{1}{n+\tfrac{1}{2}-j-k_1} \cdot
\frac{1}{n+\tfrac{1}{2}-j-k_2} \Big) .
\end{align*}
Bound each of the two terms separately. 
For the first term, for each $k$,
\begin{align*}
\sum_{j \in J} \frac{1}{(n+\tfrac{1}{2}-j-k)^2}
\leq 2 \sum_{\ell=0}^\infty \frac{1}{(\ell+\tfrac{1}{2})^2} 
=  8 \sum_{\ell=0}^\infty \frac{1}{(2\ell+1)^2}  = \pi^2 ;
\end{align*}
the last equality can be justified as follows.
Because $\sum_{\ell=1}^\infty \frac{1}{\ell^2} = \tfrac{\pi^2}{6}$,
we know that $\sum_{\ell=1}^\infty \frac{1}{(2\ell)^2} = \tfrac{\pi^2}{24}$,
so $\sum_{\ell=0}^\infty \frac{1}{(2\ell+1)^2} = \pi^2 (\tfrac{1}{6}-\tfrac{1}{24})$.
For the second term, for every $k_1 \neq k_2$,
\begin{align*}
& \Big| \sum_{j \in J} \frac{1}{n+\tfrac{1}{2}-j-k_1} \cdot
\frac{1}{n+\tfrac{1}{2}-j-k_2}\Big| \\
& =  \frac{1}{|k_2-k_1|} \cdot \Big|\Big(\sum_{i \in I}\frac{1}{n+\tfrac{1}{2}-j-k_1}\Big) -
\Big( \sum_{i \in I} \frac{1}{n+\tfrac{1}{2}-j-k_2} \Big) \Big| \\
& \leq \frac{1}{|k_2-k_1|} \cdot 2 |k_2-k_1| \Big| \frac{1}{N-\tfrac{1}{2}-n-|k_1|-|k_2|} \Big| ,
\end{align*}
where we used the fact that most terms cancel out.
The last quantity tends to zero as $N \to \infty$.
This completes the proof
because there are only $n(n-1)$ significant $k_1 \neq k_2$. 
\end{proof}

\begin{proof}[Proof of Claim~\ref{clm:Hlarge}]
For each integer $0 \leq \ell \leq n-1$,
there are $n-\ell$ pairs $(j,k)$ so that $j+k=n+1-\ell$,
and
there are $n-\ell$ pairs $(j,k)$ so that $j+k=n+1+\ell$.
So,
\begin{align*}
 \sum_{j,k} \Big|\frac{1}{n+\tfrac{1}{2}-j-k}\Big| 
& = \sum_{\ell =0}^{n-1} \Big| \frac{n-\ell}{n+\tfrac{1}{2}-(n+1-\ell)}\Big|
+ \Big| \frac{n-\ell}{n+\tfrac{1}{2}-(n+1+\ell)} \Big| \\
& = \sum_{\ell =0}^{n-1} \frac{n-\ell}{\ell-\tfrac{1}{2}}
+ \frac{n-\ell}{\ell+\tfrac{1}{2}} \\
& = n \Big( \sum_{\ell =0}^{n-1} \frac{1}{\ell-\tfrac{1}{2}}
+ \frac{1}{\ell+\tfrac{1}{2}}\Big) 
- \Big( \sum_{\ell =0}^{n-1} \frac{\ell}{\ell-\tfrac{1}{2}}
+ \frac{\ell}{\ell+\tfrac{1}{2}}\Big) \\
& \geq 2 n (\ln(n)-3) . \qedhere
\end{align*}

\end{proof}

\subsection{Insufficiency of previous hard distributions}
\label{sec:previous}

The previous works~\cite{viola2015communication,braverman2016discrepancy,ramamoorthy2015communication}
used different distributions that leads to a sub-optimal bound of $\disc(G) \leq O(\tfrac{1}{\sqrt{\ln n}})$;
this bound was proved in~\cite{braverman2016discrepancy}.
In this section, we focus on the distribution~$\eta$ defined in~\cite{braverman2016discrepancy},
and prove that this bound is, in fact, the best that can be obtained for this distribution.

We first recall the distribution $\eta$. 
Assume that $n=2^m$ for some positive integer $m$;
otherwise work with the largest power of two smaller than $n$. 
Identify $[n]$ with $\{0,1\}^m$ by identifying $j$ with the binary representation of $j-1$ (with the first coordinate being the most significant bit and the last coordinate being the least significant). 
Sample $(j,k)$ from $\eta$ as follows. First, sample $j$ uniformly at random from $\{0,1\}^m$,
and independently choose a uniformly random $i\in [m]$. 
Set $k$ to be equal to $j$ in all coordinates less than $i$, to be the opposite of $j$ in coordinate $i$, and chosen independently of $j$ and uniformly in all coordinates greater than $i$.

\begin{claim*}
$\disc_\eta(G) \geq \Omega\big(\tfrac{1}{\sqrt{\ln n}}\big).$
\end{claim*}

\begin{proof}
Define $x\in \mathbb{R}^n$ to be the indicator vector of all those $j\in \{0,1\}^m$ that have at least $\tfrac{m}{2} + \sqrt{m}$ many one entries, and $y\in \mathbb{R}^n$ to be the all-ones vector. 
Denote by $j_i$ the $i$'th coordinate of $j$.
The event $j\geq k$ holds exactly when $j_i=1$.
It follows that
\begin{align*}
\disc_\eta(G) & = \sum_{j,k} G_{j,k} \eta_{j,k} x_j \\
& = \Pr_{\eta}\left[ x_j = 1 \wedge j_i = 1\right] -  \Pr_{\mu}\left[x_j = 1 \wedge j_i = 0\right]\\
&= \Pr_{\eta}\left[x_j = 1\right]\cdot (2 \Pr_{\mu}\left[j_i = 1\ |\  x_j = 1\right] -1) .
\end{align*}
Standard estimates imply that $\Pr_{\eta}\left[x_j = 1\right] = \Omega(1).$ 
Conditioned on any such choice for $j$, the random variable $i$ is still uniform, and thus the probability that $j_i =1$ is at least $\tfrac{1}{2} + \tfrac{1}{\sqrt{m}}$. 
\end{proof}

\section{The lower bound}

%
%
%
%
%
%
%
%
%
%
%

In this section, we work with the matrix $G_{j,k} = 2_{j \geq k}-1$.
The advantage is that this $G$ is a Toeplitz matrix
$G_{j,k} = c_{j-k}$, where in the previous sections
it was a Hankel matrix $G_{j,k} = c_{j+k}$.
This distinction is meaningless in terms of discrepancy,
but turns out to be important in analysis. 

The mechanism for locating the witnesses $x,y$
relies on the existence of a certain function $\nu$
that encodes the interaction of $G$ with many witnesses
(via the Fourier transform). 
A {\em complex measure} is a bounded, Borel, absolutely continuous $\nu : [0,1) \to \C$.
Its Fourier transform is defined to be
$$\Z \ni \ell \mapsto \hat{\nu}(\ell) = \int_0^1 e^{2 \pi i \ell t} \nu(t) dt \in \C.$$
Its norm is defined to be
$$\|\nu\| = \sup_{f : [0,1) \to \C , \|f\|_\infty =1} \Big| \int_0^1 \nu(t) f(t) dt \Big|.$$
The following is implicit in Bennett's work.

\begin{lemma}[implicit in \cite{Bennett_1977}]
\label{lem:B}
There is a complex measure $\nu$ so that
for all $\ell$ so that $|\ell|\leq n$,
\begin{align}
\label{eq:muhat}
\hat{\nu}(\ell) =
\begin{cases}
1 & \ell \geq 0\\
0 & \ell < 0
\end{cases}
\end{align}
and so that
\begin{align}
\label{eqn:normnu}
\|\nu\| \leq \frac{\ln n}{\pi}+2.
\end{align}
\end{lemma}

\begin{proof}[Proof of Lemma~\ref{lem:B}]
The references here are in Bennett's work~\cite{Bennett_1977}.
Theorem~8.1 states the existence of $\nu$ satisfying
$\hat{\nu}(\ell) = M_{j+\ell,j}$
and $\|\nu\|  = \|M\|_{(\infty,1)}$ for all Toeplitz matrices that are multipliers.
Let $M$ be the $n \times n$ Toeplitz matrix $M_{j,k} = 1_{j \geq k}$.
The proof of Corollary 8.5 shows that $M$ is a multiplier 
and also proves the bound on $\|\nu\|=\|M\|_{(\infty,1)}$ stated in~\eqref{eqn:normnu}.
In the proof of Corollary 8.3,
it is also proved that $\nu$ is an absolutely continuous $L^1$-function.
\end{proof}

\begin{proof}[Proof of lower bound]
Let $\mu = \mu_{j,k}$ be a distribution. 
Without loss of generality, we can assume $\Pr_\mu[j \geq k] \geq \tfrac{1}{2}$.
By Lemma~\ref{lem:B},
\begin{align*}
\frac{1}{2}
& \leq \Big| \sum_{j,k} G_{j,k} \mu_{j,k} \hat{\nu}(j-k)  \Big| \\
& = \Big| \sum_{j,k} G_{j,k} \mu_{j,k} \int \nu(t) e^{2 \pi i (j-k)t} dt  \Big| \\
& = \Big| \int \nu(t) \sum_{j,k} G_{j,k} \mu_{j,k} e^{2 \pi i j t}
e^{-2 \pi i k t}  dt \Big| \\
& \leq \|\nu\| \cdot \disc_\mu(G) . \qedhere
\end{align*}
\end{proof}

\bibliographystyle{amsplain}
\bibliography{discref.bib}

\end{document}